\documentclass{article}

\usepackage{algorithm}
\usepackage[noend]{algpseudocode}
\usepackage{amsmath,amssymb,amsthm}
\usepackage{bm}
\usepackage{bbm}
\usepackage{bbold}
\usepackage{booktabs}
\usepackage{caption}
\usepackage{enumerate}
\usepackage{float}
\usepackage[colorlinks,citecolor=blue,linkcolor=BrickRed]{hyperref}
\usepackage{cleveref}
\usepackage[utf8]{inputenc}
\usepackage{parskip}
\usepackage{subcaption}
\usepackage{times}
\usepackage[normalem]{ulem}
\usepackage[usenames,dvipsnames]{xcolor}
\usepackage{xspace}

\newtheorem{lemma}{Lemma}

\newcounter{note}[section]

\newcommand{\defeq}{\overset{\text{def}}{=}}
\newcommand{\E}{\mathop{{}\mathbb{E}}}

\newcommand{\st}{\text{\xspace s.t. \xspace}}

\newcommand{\mcD}{\mathcal{D}}

\allowdisplaybreaks

\usepackage{fullpage}
\usepackage{authblk}
\usepackage{tikz}
\usepackage{pgfplots}
\usepackage{ifthen}
\usepackage{thmtools}
\usepackage{thm-restate}
\usepackage{MnSymbol}
\usepackage{caption}
\usepackage[thinc]{esdiff}

\usepackage[numbers,sort&compress]{natbib}
\usetikzlibrary{positioning}
\usetikzlibrary{arrows}
\usetikzlibrary{decorations.markings}

\usepackage{fullpage}
\usepackage{tikz}
\usetikzlibrary{calc}

\newcommand{\mscs}{\textsc{WMSCSS}\xspace}

\newcommand{\pmscs}{\ensuremath{P_\text{MSCS}}}

\title{An Optimal Rounding for Half-Integral Weighted Minimum Strongly Connected Spanning Subgraph}

\author[1]{D Ellis Hershkowitz\thanks{dhershko@cs.cmu.edu}}
\author[1]{Gregory Kehne\thanks{gkehne@andrew.cmu.edu}}
\author[1]{R.\ Ravi\thanks{ravi@andrew.cmu.edu}}
\affil[1]{Carnegie Mellon University}

\begin{document}
	
	\maketitle
	\begin{abstract}
		In the weighted minimum strongly connected spanning subgraph (\mscs) problem we must purchase a minimum-cost strongly connected spanning subgraph of a digraph. We show that half-integral linear program (LP) solutions for \mscs can be efficiently rounded to integral solutions at a multiplicative $1.5$ cost. This rounding matches a known $1.5$ integrality gap lower bound for a half-integral instance. More generally, we show that LP solutions whose non-zero entries are at least a value $f > 0$ can be rounded at a multiplicative cost of $2 - f$.
	\end{abstract}

		\section{Introduction}
	The weighted minimum strongly connected spanning subgraph (\mscs) problem is arguably the simplest NP-hard connectivity problem on directed graphs. In \mscs we are given a strongly connected\footnote{A strongly connected digraph is one in which every node has a directed path to every other node.} digraph $D = (V, A)$ with weight function $w : A \to \mathbb{R}^+$. Our goal is to purchase a strongly connected spanning subgraph $H = (V, A')$ of $G$ of minimum cost, where the cost of $H$ is $w(A') := \sum_{a \in A'} w(a)$. The simplicity of \mscs has lent itself to several applications in network design and computational biology \cite{moyles1967algorithm,khuller1995approximating,vincent2005transitive,aditya2013algorithmic}.
	
	Unfortunately \mscs is NP-hard \cite{johnson1979computers}. Even worse, it has been shown to be MaxSNP hard, meaning that it admits no polynomial-time approximation scheme assuming $\text{P} \neq \text{NP}$ \cite{khuller1995approximating}.
	
	Fortunately \mscs admits a simple 2-approximation due to \citet{frederickson1981approximation} which employs min-cost arborescences. Given digraph $D=(V,A)$ and root $r\in V$, an $r$-in-arborescence of $D$ is a spanning subgraph $I = (V,A')$ such that every node $v \neq r$ has exactly one path to $r$ along edges in $A'$. An $r$-out-arboresecence $O$ is defined analogously with paths from $r$ to $v$. Minimum-weight $r$-in- and $r$-out-arborescences can be computed in polynomial time \cite{chu1965shortest,edmonds1967optimum}. The mentioned $2$-approximation simply fixes an arbitrary root $r$ and then takes the union of a min-cost $r$-in-arborescence and a min-cost $r$-out-arborescence. Since every node has a path to and from $r$ in their union, the result is a feasible \mscs solution. Moreover, the result is a $2$-approximation since the optimal \mscs is a strongly connected spanning subgraph and so contains a feasible $r$-in- and $r$-out-arborescence as a subgraph for any choice of $r$. Remarkably, this $2$-approximation has remained the best-known polynomial-time approximation for \mscs for almost 40 years.
	
	
	Thus, \mscs falls into a class of combinatorial optimization problems which admit simple, polynomial-time algorithms whose constant approximation ratios have not been improved in many decades. Notable other examples include the Traveling Salesman Problem (TSP) for which Christofides' simple $1.5$-approximation \cite{christofides1976worst} has remained the best polynomial-time approximation since 1976 and the Weighted Tree Augmentation Problem (WTAP) for which the best known polynomial-time approximation ratio is $2$ as established by \citet{frederickson1981approximation} in 1981 in the same work that gave a $2$ approximation for \mscs. Given the apparent difficulty in improving these bounds, a great deal of work has focused on improving the approximation ratios of algorithms for special cases of these problems \cite{gharan2011randomized,momke2011approximating, kortsarz2015simplified, khuller1995approximating,vetta2001approximating,berman2009approximating,grimmer2018dual,zhao2003linear, fiorini20173}.

	A recently fruitful such special case has been the assumption that solutions to the relevant linear program (LP) are half-integral---that is, each coordinate of an optimal solution is assumed to lie in $\{0, \frac{1}{2}, 1\}$. Notably, \citet{cheriyan19992} showed that WTAP admits a $4/3$-approximation if the relevant LP is half-integral and \citet{iglesias2017coloring} generalized this by showing that a $2/(1+f)$-approximation is possible for WTAP if non-zero LP values are assumed to be at least $f > 0$. Similarly, a recent breakthrough of \citet{karlin2019improved} showed that a $\approx 1.49993$ approximation is possible for TSP if the LP solution is assumed to be half-integral. Studying such special cases offers the opportunity to develop tools and to help delineate lower and upper bounds for the general case. For example, TSP instances with half-integral optimal LP solutions are conjectured to be the hardest TSP instances to approximate \cite{schalekamp20142} and so the work of \citet{karlin2019improved} was taken as evidence that Christofides' algorithm does not, in fact, attain the best constant approximation among all polynomial-time algorithms---a suspicion recently confirmed by \citet{karlin2020slightly}. 


\subsection{Our Contributions}

In this work we take this approach to \mscs. Specifically, we adapt a deterministic algorithm of \citet{laekhanukit2012rounding} to show that half-integral solutions for the \mscs LP can be deterministically rounded in polynomial time at a multiplicative cost of $1.5$. In this LP, we enforce that every non-trivial cut has at least one purchased edge leaving. We use $\delta^+(S)$ to denote the set of arcs leaving $S \subset V$.
\begin{align*}\tag{\mscs LP}\label{WMSCSLP}
\min \: w(x) \defeq w^T \cdot x \\
\st \qquad x(\delta^+(S)) &\geq 1 & \forall \emptyset \subset S \subset V\\
x_a &\geq 0 &\forall a \in A
\end{align*}

\citet{laekhanukit2012rounding} gave a family of half-integral instances of \mscs for which the integrality gap of the \ref{WMSCSLP} is bounded below by $1.5-\epsilon$ for any $\epsilon > 0$, and so our upper bound of $1.5$ is tight.

More generally, we show how to round any LP solution with non-zero entries bounded below by $f>0$ at a multiplicative cost of $2-f$. 

Our result for half-integral solutions may be seen as adding to a growing body of evidence that a polynomial-time $1.5$-approximation is both achievable and the best possible for \mscs. Prior evidence includes a series of works that culminated in a $1.5$-approximation for the unit-cost case of \mscs \cite{khuller1995approximating,vetta2001approximating,berman2009approximating} and the aforementioned $1.5$ integrality gap lower bound, which is the best known integrality gap lower bound for \ref{WMSCSLP}. Since the best known integrality gap lower bound is attained by a half-integral solution, it seems possible that for \mscs, like for TSP, the hardest-to-approximate instances may be half-integral.

Our rounding algorithm will be a degenerate form of one proposed by \citet{laekhanukit2012rounding}. In particular, \citet{laekhanukit2012rounding} proposed an algorithm to study the $k$-arc connected subgraph problem for $k \geq 2$ on \emph{unit-cost} graphs. This is the $k$-arc connected generalization of the unit-cost \mscs. We make use of the same algorithm but use it with $k=1$ \emph{on arbitrary cost graphs}; thus, the setting in which we apply this algorithm is somewhat different from that of prior work.

	\section{Rounding}
We review the algorithm of \citet{laekhanukit2012rounding} for $k=1$ before presenting our new analysis.
\subsection{Algorithm of \citet{laekhanukit2012rounding}}
The algorithm makes use of the $r$-in- and $r$-out-arborescence LPs, defined as follows.
\begin{align*}\tag{In-Arborescence LP}\label{LPArbIn}
\min \: w(x) \\
\st \quad x(\delta^+(S)) &\geq 1 & \forall \emptyset \subset S \subset V \st r \not \in S\\
x_a &\geq 0 &\forall a \in A
\end{align*}
The $r$-out-arborscence LP is defined symmetrically; in particular ``$\delta^+(S)$'' is replaced with ``$\delta^-(S) $'' (the set of edges entering the set $S$).

Given an $x$ which is feasible for the arborescence LP, it is known how to efficiently sample arborescences in a manner consistent the marginals defined by $x$, as summarized by the following claim. In the following, $\mcD_\text{in}$ is a distribution over subgraphs of $D$ each of which contains an in-arborescence and $\mcD_\text{out}$ is a distribution over subgraphs of $D$ each of which contains an out-arborescence.\footnote{In the algorithm originally presented by $\mcD_\text{in}$ \citet{laekhanukit2012rounding} the lemma has $I$ \emph{equal} to an in-arborescence and $O$ \emph{equal} to an out-arborescence and $\Pr(a \in I) \leq x_a$ and $\Pr(a \in O) \leq x_a$. It is easy to see that the claimed lemma immediately follows by adding extra arcs to $I$ and $O$ to increase the probability arc $a$ is sampled to be exactly $x_a$ whenever that is not the case.}

\begin{lemma}[\cite{frank1979covering,carr2002randomized,edmonds1973edge,gabow1995matroid}]\label{thm:sampling}
	Given an $x$ feasible for the $r$-in- and $r$-out-arborescence LP's, we can, in polynomial time, independently sample sub-graphs $I\sim \mcD_\text{in}$ and $O \sim \mcD_{\text{out}}$ such that $\Pr(a \in I) = \Pr(a \in O) = x_a$ for every $a \in A$. Moreover, the size of the support of $\mcD_\text{in}$ and $\mcD_\text{out}$ is polynomial in the size of the graph and is computable in polynomial-time.
\end{lemma}

The algorithm of \citet{laekhanukit2012rounding} which we adapt for our case is as follows: given an $x$ feasible for \ref{WMSCSLP} for an arbitrary root $r$, enumerate the support of $\mcD_\text{in}$ and $\mcD_\text{out}$ as in \Cref{thm:sampling} to get digraphs $I_1, \ldots I_{\alpha}$ and $O_1, \ldots O_{\beta}$. Let $\chi_{ij}$ be $I_i \cup O_j$. Return the $\chi_{ij}$ of minimum weight.\footnote{Here \citet{laekhanukit2012rounding} actually used the optimal $x$. Additionally, \citet{laekhanukit2012rounding} stated their algorithm as choosing uniformly at random from among the $\chi_{ij}$ but then remarked that it can be derandomized; we give the derandomized version.}

\subsection{Rounding LP Solutions}
We now apply this algorithm to LP solutions whose non-zero entries are at least $f$. Let $\pmscs$ be the polytope corresponding to \ref{WMSCSLP} and let $P_{f} := \{x \in \pmscs :  x_a \not \in (0, f) \text{ for all } a \in A\}$ be all feasible points whose non-zero entries are at least $f$.

\begin{restatable}{theorem}{algRand}\label{thm:algRand}
	Given an $x \in P_f$, the above deterministic algorithm outputs an integral \mscs solution $\hat{x}$ such that $w(\hat{x}) \leq (2-f)\cdot w(x)$ in polynomial time.
\end{restatable}
\begin{proof}
	Let $\chi$ be the digraph returned by the algorithm. Since each $O_i$ and $I_j$ contains an $r$-out and $r$-in-arborescence, the characteristic vector corresponding to $\chi$ is a feasible integral solution to \ref{WMSCSLP}. Moreover, a polynomial runtime follows immediately from \Cref{thm:sampling} and the fact that there are only polynomially many $\chi_{ij}$'s.
	
	Thus, let us bound the cost of $\chi$. By an averaging argument, it suffices to upper bound the expected cost of $\hat{\chi} = I \cup O$ where $I  \sim \mcD_{\text{in}}$ and $O \sim \mcD_{\text{out}}$ since $\chi$ is in the support of this distribution, and the minimum cost member in the collection has cost at most the average.
	
	Consider a fixed arc $a$. By the inclusion-exclusion principle, the fact that $a \in I$ is independent of whether $a \in O$, and \Cref{thm:sampling}, the probability that $a$ is in $\hat{\chi}$ is
	\[
	\Pr(a \in \hat\chi) =\Pr(a \in I) + \Pr(a \in O) - \Pr(a \in I)\Pr(a \in O) = 2x_a - x_a^2.
	\]
	
	Thus, by linearity of expectation we have
	\begin{align*}
	\E[w(\hat{\chi})] &= \sum_a w(a) \cdot (2x_a - x_a^2)\\
		&\leq \sum_{a} w(a) \cdot (2x_a - x_a \cdot f) & (\text{Since } x_a > 0 \implies x_a \geq f \ \  \forall x_a \in P_f)\\
	& = (2-f) \cdot w(x).
	\end{align*}
\end{proof}	

As a corollary of \Cref{thm:algRand} we recover our $1.5$-cost rounding for half-integral solutions:
\begin{restatable}{corollary}{halfCor}\label{cor:half}
	There is a deterministic polynomial-time algorithm which, given an $x \in P_{1/2}$, outputs a feasible integral $\hat{x}$ such that $w(\hat{x}) \leq 1.5 \cdot w(x)$.
\end{restatable}

\section{Acknowledgments}
D Ellis Hershkowitz supported in part by NSF grants CCF-1527110, CCF-1618280, CCF-1814603,CCF-1910588, NSF CAREER award CCF-1750808 and a Sloan Research Fellowship. R Ravi supported in part by the U.S.\ Office of Naval Research award N00014-18-1-2099.

\bibliographystyle{plainnat}
\bibliography{abb,main}

\end{document}